%% file: generalIB.tex
\documentclass[conference,9pt,final]{IEEEtran}

\IEEEoverridecommandlockouts

\newcommand*{\uselmodern}{true}

\input{packages}

\input{custom_environments}

\input{custom_commands}

\input{custom_math}

\input{custom_probability}

\input{custom_variables}

\input{custom_versions}

\input{tikzsettings}

\newlist{properties}{enumerate}{1}
\setlist[properties]{label=\textnormal{(\roman*)}}

\crefname{property}{Prop.}{Props.}
\Crefname{property}{Prop.}{Props.}

\crefname{definition}{Def.}{Defs.}
\Crefname{definition}{Def.}{Defs.}
\crefname{lemma}{Lem.}{Lems.}
\Crefname{lemma}{Lem.}{Lems.}
\crefname{theorem}{Thm.}{Thms.}
\Crefname{theorem}{Thm.}{Thms.}
\crefname{remark}{Rmk.}{Rmks.}
\Crefname{remark}{Rmk.}{Rmks.}
\crefname{corollary}{Cor.}{Cors.}
\Crefname{corollary}{Cor.}{Cors.}
\crefname{figure}{Fig.}{Figs.}
\Crefname{figure}{Fig.}{Figs.}

\def\nil{\varnothing}

\def\diam{\mathrm{diam}}

\newcommand{\ps}[1][\empty]{%
  \ifthenelse{\equal{#1}{\empty}}{%
    \Omega}{%
    \SSS_{\rv #1}}%
}
\newcommand{\rng}[1][\empty]{%
  \ifthenelse{\equal{#1}{\empty}}{%
    \MMM}{%
    \MMM_{\rv #1}}%
}

\newcommand{\fkt}[1][\empty]{%
  \ifthenelse{\equal{#1}{\empty}}{%
    a}{%
    a_{\rv #1}}%
}
\newcommand{\fktI}[1][\empty]{%
  \ifthenelse{\equal{#1}{\empty}}{%
    b}{%
    b_{\rv #1}}%
}
\newcommand{\fktII}[1][\empty]{%
  \ifthenelse{\equal{#1}{\empty}}{%
    c}{%
    c_{\rv #1}}%
}
\newcommand{\sa}[1][\empty]{%
  \ifthenelse{\equal{#1}{\empty}}{%
    \Sigma}{%
    \AAA_{\rv #1}}%
}

\tikzstyle{rvsketch} = [row sep=10mm,
column sep=15mm,
>=latex,
ampersand replacement=\&,
line width=1pt,
baseline=2mm,
block/.style={draw, rectangle, minimum height=2em, rounded corners=2mm},
dec/.style={draw, rectangle, minimum height=2em, rounded corner=2mm},
symb/.style={outer sep=1.5mm},
prob/.style={outer sep=-1mm},
prob1/.style={outer sep=-.3mm},
]

\allowdisplaybreaks

\begin{document}

\setlength\textfloatsep{2pt plus 3pt}

\title{\vspace*{-.5cm}
Information Bottleneck on General Alphabets}

\author{\IEEEauthorblockN{Georg~Pichler, Günther~Koliander}
  \thanks{G.\ Pichler is with the Institute of Telecommunications, Technische Universität Wien, Vienna, Austria}
  \thanks{G.\ Koliander is with the Acoustics Research Institute, Austrian Academy of Sciences, Vienna, Austria}
  \thanks{\scriptsize Funding by WWTF Grants MA16-053, ICT15-119, and NXT17-013.}
}

\makeatletter
\patchcmd{\@maketitle}
  {\addvspace{0.5\baselineskip}\egroup}
  {\addvspace{-\baselineskip}\egroup}
  {}
  {}
\makeatother

\maketitle

\begin{abstract}
  We prove rigorously a source coding theorem that can probably be considered folklore, a generalization to arbitrary alphabets of a problem motivated by the Information Bottleneck method. For general random variables $(\rv y, \rv x)$, we show essentially that for some $n \in \NN$, a function $f$ with rate limit $\log\card{f} \le nR$ and $\mutInf{\rv y^n}{f(\rv x^n)} \ge nS$ exists if and only if there is a random variable $\rv u$ such that the Markov chain $\rv y \mkv \rv x \mkv \rv u$ holds, $\mutInf{\rv u}{\rv x} \le R$ and $\mutInf{\rv u}{\rv y} \ge S$.
  The proof relies on the well established discrete case and showcases a technique for lifting discrete coding theorems to arbitrary alphabets.
\end{abstract}

\IEEEpeerreviewmaketitle

\section{Introduction}
Since its inception \cite{Tishby2000Information}, the \emph{Information Bottleneck} (IB) method became a widely applied tool, especially in the context of machine learning problems.
It has been successfully applied to various problems in machine learning \cite{Slonim2000Document}, computer vision \cite{Gordon2003Applying}, and communications \cite{Zeitler2009quantizer,Zeitler2012Low,Winkelbauer2012Joint}.
Furthermore, it is a valuable tool for channel output compression in a communication system \cite{Winkelbauer2014rate,Winkelbauer2015quantization}.

In the underlying information-theoretic problem, we define a pair $(S,R) \in \RR^2$ to be \emph{achievable} for the two arbitrary random sources $(\rv y, \rv x)$, if there exists a function $f$ with rate limited range $\frac{1}{n} \log\card{f} \le R$ and $\mutInf{\rvt y}{f(\rvt x)} \ge nS$, where $(\rvt y, \rvt x)$ are $n$ independent and identically distributed (\iid) copies of $(\rv y, \rv x)$. 

While this Shannon-theoretic problem and variants thereof were also considered (\eg, \cite{Pichler2015Distributed,Courtade2014Multiterminal}), a large part of the literature is aimed at studying the IB function
\begin{align}
    S_{\mathrm{IB}}(R) 
  = \sup_{\substack{\rv u \;:\; \mutInf{\rv u}{\rv x} \le R \\
  \rv y \mkv \rv x \mkv \rv u}}
  \mutInf{\rv u}{\rv y} \label{eq:IB-function}
\end{align}
in different contexts. In particular, several works (\eg, \cite{Tishby2000Information,Slonim2000Document,Chechik2005Information,Slonim2000Agglomerative,Kurkoski2017Relationship}) intend to compute a probability distribution that achieves the supremum in \cref{eq:IB-function}. The resulting distribution is then used as a building block in numerical algorithms, \eg, for document clustering \cite{Slonim2000Document} or dimensionality reduction \cite{Chechik2005Information}.

In the discrete case, $S_{\mathrm{IB}}(R)$ is equal to the maximum of all $S$ such that $(S,R)$ is in the \emph{achievable region} (closure of the set of all achievable pairs). This statement has been re-proven many times in different contexts \cite{Westover2008Achievable,Courtade2014Multiterminal,Ahlswede1986Hypothesis,Han1987Hypothesis}.
In this note, we prove a theorem, which can probably be considered folklore, extending this result from discrete to arbitrary random variables.
Formally speaking, using the definitions in \cite{Gray1990Entropy}, we prove that a pair $(S,R)$ is in the achievable region of an arbitrary source $(\rv y, \rv x)$ if and only if, for every $\eps > 0$, there exists a random variable $\rv u$ with $\rv y \mkv \rv x \mkv \rv u$, $\mutInf{\rv x}{\rv u} \le R + \eps$, and $\mutInf{\rv y}{\rv u} \ge S - \eps$.
This provides a single-letter solution to the information-theoretic problem behind the information bottleneck method for arbitrary random sources and in particular it shows, that the information bottleneck for Gaussian random variables \cite{Chechik2005Information} is indeed the solution to a Shannon-theoretic problem.

The proof relies on the discrete case. Thus, the techniques employed could be useful for lifting other discrete coding theorems to the case of arbitrary alphabets.

\section{Main Result}
\label{sec:main-result}

Let $\rv y$ and $\rv x$ be random variables with arbitrary alphabets $\ps[y]$ and $\ps[x]$, respectively. The bold-faced random vectors $\rvt y$ and $\rvt x$ are $n$ \iid copies of $\rv y$ and $\rv x$, respectively. We then have the following definitions.

\begin{definition}
  \label{def:achievable}
  A pair $(S, R) \in \RR^2$ is \emph{achievable} if for some $n \in \NN$ there exists a measurable function $f\colon \ps[x]^n \to \MMM$ for some finite set $\MMM$ with bounded cardinality $\frac{1}{n} \log\card{\MMM} \le R$ and
  \begin{align}
    \label{eq:achievable}
    \frac{1}{n} \mutInf[big]{\rvt y}{f(\rvt x)} \ge S .
  \end{align}
  The set of all achievable pairs is denoted $\RRR \subseteq \RR^2$.
\end{definition}

\begin{definition}
  \label{def:ib-achievable}
  A pair $(S, R) \in \RR^2$ is \emph{IB-achievable} if there exists an additional random variable $\rv u$ with arbitrary alphabet $\ps[u]$, satisfying $\rv y \mkv \rv x \mkv \rv u$ and
  \begin{align}
    R &\ge \mutInf[normal]{\rv x}{\rv u} , \label{eq:Rcondition} \\
    S &\le \mutInf[normal]{\rv y}{\rv u} . \label{eq:Scondition} 
  \end{align}
  The set of all IB-achievable pairs is denoted $\RRRib \subseteq \RR^2$.
\end{definition}

In what follows, we will prove the following theorem.
\begin{theorem}
  \label{thm:main}
  The equality $\ol\RRRib = \ol\RRR$ holds.
\end{theorem}

\section{Preliminaries}
\label{sec:preliminaries}
When introducing a function, we implicitly assume it to be measurable \wrt the appropriate $\sigma$-algebras.
The $\sigma$-algebra associated with a finite set is its power set and the $\sigma$-algebra associated with $\RR$ is the Borel $\sigma$-algebra.
The symbol $\nil$ is used for the empty set and for a constant random variable.
When there is no possibility for confusion, we will not distinguish between a single-element set and its element, \eg, we write $x$ instead of $\{x\}$ and $\ind{x}{}$ for the indicator function of $\{x\}$.
We use $A \vartriangle B \defas (A\setminus B) \cup (B\setminus A)$ to denote the symmetric set difference.

Let $(\ps, \sa, \mu)$ be a probability space.
A random variable $\rv x\colon \ps \to \ps[x]$ takes values in the measurable space $(\ps[x], \sa[x])$.
The push-forward probability measure $\mu_{\rv x}\colon \sa[x] \to [0,1]$ is defined by $\mu_{\rv x}(A) = \mu\big(\rv x^{-1}(A)\big)$ for all $A \in \sa[x]$.
We will state most results in terms of push-forward measures and usually ignore the background probability space.
When multiple random variables are defined, we implicitly assume the push-forward measures to be consistent in the sense that, \eg, $\mu_{\rv x}(A) = \mu_{\rv x \rv y}(A \times \ps[y])$ for all $A \in \sa[x]$.

For $n \in \NN$ let $\ps^n$ denote the $n$-fold Cartesian product of $(\ps, \sa, \mu)$. A bold-faced random vector, \eg, $\rvt x$, defined on $\ps^n$, is an $n$-fold copy of $\rv x$, \ie, $\rvt x = \rv x^n$. Accordingly, the corresponding push-forward measure, \eg, $\mu_{\rvt x}$ is the $n$-fold product measure.

For a random variable $\rv x$ let $\fkt[x]$, $\fktI[x]$, and $\fktII[x]$ denote arbitrary functions on $\ps[x]$, each with finite range. We will use the symbol $\rng[x]$ to denote the range of $\fkt[x]$, \ie, $\fkt[x]\colon \ps[x] \to \rng[x]$.

\begin{definition}[{\cite[Def.~8.11]{Klenke2013Probability}}]
  \label{def:cond-expectation}
  The \emph{conditional expectation} of a random variable $\rv x$ with $\ps[x] = \RR$, given a random variable $\rv y$, is a random variable $\condExp{\rv x}{\rv y}$ such that
  \begin{enumerate}
  \item $\condExp{\rv x}{\rv y}$ is $\sigma(\rv y)$-measurable, and
  \item \label{itm:cond-expectation:tower} for all $A \in \sigma(\rv y)$, we have 
    $\Exp[][big]{\ind{A}{}\condExp{\rv x}{\rv y}} = \Exp{\ind{A}{}\rv x}$ .
  \end{enumerate}
  The \emph{conditional probability} of an event $B \in \sa$ given $\rv y$ is defined as $\Pcond{B}{\rv y} \defas \condExp{\ind{B}{}}{\rv y}$.
\end{definition}
\noindent
The conditional expectation and therefore also the conditional probability exists and is unique up to equality almost surely by \cite[Thm.~8.12]{Klenke2013Probability}.
Furthermore, if $(\ps[x], \sa[x])$ is a standard space \cite[Sec.~1.5]{Gray1990Entropy}, there even exists a \emph{regular conditional distribution} of $\rv x$ given $\rv y$ \cite[Thm.~8.37]{Klenke2013Probability}.
\begin{definition}
  \label{def:regular-conditional-distribution}
  For two random variables $\rv x$ and $\rv y$ a \emph{regular conditional distribution} of $\rv x$ given $\rv y$ is a function $\kappa_{\rv x|\rv y}\colon \Omega \times \sa[x] \to [0,1]$ such that 
  \begin{enumerate}
  \item for every $\omega \in \ps$, the set function $\kappa_{\rv x|\rv y}(\omega) \defas \kappa_{\rv x|\rv y}(\omega; \wc)$ is a probability measure on $(\ps[x], \sa[x])$.
  \item for every set $A \in \sa[x]$, the function $\kappa_{\rv x|\rv y}(\wc;A)$ is $\sigma(\rv y)$-measurable.
  \item \label{itm:regular-conditional-distribution:def} for $\mu$-\Ae $\omega \in \ps$ and all $A \in \sa[x]$, we have
    $\kappa_{\rv x|\rv y}(\omega; A) =\Pcond{\rv x^{-1}(A)}{\rv y}(\omega)$ (\cf \cref{def:cond-expectation}).
  \end{enumerate}
\end{definition}
\noindent
Note, in particular, that finite spaces are standard spaces.
\begin{remark}
  \label{rmk:conditional-probability}
  If the random variable $\rv y$ is discrete, then $\kappa_{\rv x |\rv y}$ reduces to conditioning given events $\rv y = y$ for $y \in \ps[y]$, \ie, $\kappa_{\rv x |\rv y}(\omega; A) = \frac{\mu_{\rv x \rv y}(A \times \rv y(\omega))}{\mu_{\rv y}(\rv y(\omega))}$ (\cf \cite[Lem.~8.10]{Klenke2013Probability}).
\end{remark}
We use the following definitions and results from \cite{Gray1990Entropy,Klenke2013Probability}.
\begin{definition}
  \label{def:mut-inf}
  For random variables $\rv x$ and $\rv y$ with $|\ps[x]| < \infty$ the \emph{conditional entropy} is defined as \cite[Sec.~5.5]{Gray1990Entropy}
  \begin{align}
    \condEnt{\rv x}{\rv y} \defas \int \ent{\kappa_{\rv x|\rv y}} \;d\mu ,
  \end{align}
  where $\ent{\wc}$ denotes discrete entropy on $\ps[x]$.
  For arbitrary random variables $\rv x$, $\rv y$, and $\rv z$ the \emph{conditional mutual information} is defined as \cite[Lem.~5.5.7]{Gray1990Entropy}%
  \begin{align}
    &\condMutInf{\rv x}{\rv y}{\rv z} \defas \sup_{\fkt[x], \fkt[y]} \int \DKL{\kappa_{\fkt[x](\rv x) \fkt[y](\rv y) | \rv z}}{\kappa_{\fkt[x](\rv x) | \rv z} \times \kappa_{\fkt[y](\rv y) | \rv z}} \;d\mu \nonumber\\*[-4mm]
    \label{eq:def-cond-mut-inf} \\
    &\;= \sup_{\fkt[x], \fkt[y]} \big[ \condEnt{\fkt[x](\rv x)}{\rv z} + \condEnt{\fkt[y](\rv y)}{\rv z}  - \condEnt{\fkt[x](\rv x) \fkt[y](\rv y)}{\rv z} \big] , \label{eq:def-cond-mut-inf2}
  \end{align}
  where $\DKL{\wc}{\wc}$ denotes Kullback-Leibler divergence \cite[Sec.~2.3]{Gray1990Entropy} and the supremum is taken over all $\fkt[x]$ and $\fkt[y]$ with finite range.
  The \emph{mutual information} is given by \cite[Lem.~5.5.1]{Gray1990Entropy} $\mutInf{\rv x}{\rv y} \defas \condMutInf{\rv x}{\rv y}{\nil}$.
\end{definition}

\begin{definition}[{\cite[Def.~12.20]{Klenke2013Probability}}]
  \label{def:cond-independence}
  For arbitrary random variables $\rv x$, $\rv y$, and $\rv z$, the Markov chain $\rv x \mkv \rv y \mkv \rv z$ holds if, for any $A \in \sa[x]$, $B \in \sa[z]$, the following holds $\mu$-\Ae:
  \begin{align}
    \Pcond{\rv x^{-1}(A) \cap \rv z^{-1}(B)}{\rv y} = \Pcond{\rv x^{-1}(A)}{\rv y}
     \Pcond{\rv z^{-1}(B)}{\rv y} .
  \end{align}
\end{definition}
\noindent
In the following, we collect some properties of these definitions.%
\begin{lemma}
  \label{lem:properties}
  For random variables $\rv x$, $\rv y$, and $\rv z$ the following properties hold:
  \begin{properties}
  \item \label[property]{itm:mut-inf-mkv} $\condMutInf{\rv x}{\rv y}{\rv z} \ge 0$ with equality if and only if $\rv x \mkv \rv z \mkv \rv y$.
  \item \label[property]{itm:discrete-cond-ent} For discrete $\rv x$, \ie, $\card{\ps[x]} < \infty$, we have
    $\mutInf{\rv x}{\rv y} = \ent{\rv x} - \condEnt{\rv x}{\rv y}$.
  \item \label[property]{itm:kolmogoroff-formula} $\mutInf{\rv x}{\rv y \rv z} = \mutInf{\rv x}{\rv z} + \condMutInf{\rv x}{\rv y}{\rv z}$.
  \item \label[property]{itm:data-processing} If $\rv x \mkv \rv y \mkv \rv z$, then $\mutInf{\rv x}{\rv y} \ge \mutInf{\rv x}{\rv z}$.
  \end{properties}
\end{lemma}
\begin{proof}
  \labelcref{itm:mut-inf-mkv}:
  The claim $\condMutInf{\rv x}{\rv y}{\rv z} \ge 0$ follows directly from \cref{eq:def-cond-mut-inf} and the non-negativity of divergence.
  
  Assume that $\rv x \mkv \rv z \mkv \rv y$, \ie, $\Pcond{\rv x^{-1}(A) \cap \rv y^{-1}(B)}{\rv z} = \Pcond{\rv x^{-1}(A)}{\rv z} \Pcond{\rv y^{-1}(B)}{\rv z}$ almost everywhere. Let $\fkt[x]\colon \ps[x] \to \rng[x]$ and $\fkt[y]\colon \ps[y] \to \rng[y]$ be functions with finite range. Pick two arbitrary sets $A \subseteq \rng[x]$, $B \subseteq \rng[y]$ and we obtain $\mu$-\Ae
  \begin{align}
    &\kappa_{\fkt[x](\rv x) \fkt[y](\rv y)| \rv z}(\wc; A \times B) \eqnl
    &\qquad= \Pcond{\rv x^{-1}(\fkt[x]^{-1}(A)) \cap \rv y^{-1}(\fkt[y]^{-1}(B))}{\rv z} \label{eq:properties:def1} \\
    &\qquad= \Pcond{\rv x^{-1}(\fkt[x]^{-1}(A))}{\rv z} \Pcond{\rv y^{-1}(\fkt[y]^{-1}(B))}{\rv z} \\
    &\qquad= \kappa_{\fkt[x](\rv x)|\rv z}(\wc; A) \kappa_{\fkt[y](\rv y)|\rv z}(\wc;B) , \label{eq:properties:def2}
  \end{align}
  where \cref{eq:properties:def1,eq:properties:def2} follow from \scref{def:regular-conditional-distribution}{itm:regular-conditional-distribution:def}.
  This proves that $\mu$-\Ae the equality of measures $\kappa_{\fkt[x](\rv x) \fkt[y](\rv y) | \rv z} = \kappa_{\fkt[x](\rv x)|\rv z} \times \kappa_{\fkt[y](\rv y)|\rv z}$ holds. By the properties of Kullback-Leibler divergence \cite[Thm.~2.3.1]{Gray1990Entropy} we have $\condMutInf{\rv x}{\rv y}{\rv z} = 0$ due to \cref{eq:def-cond-mut-inf}.

  On the other hand, assume $\condMutInf{\rv x}{\rv y}{\rv z} = 0$ and choose arbitrary sets $A \in \sa[x]$ and $B \in \sa[y]$. We define $\fkt[x] \defas \ind{A}{}$, $\fkt[y] \defas \ind{B}{}$, $\hrv x \defas \fkt[x](\rv x)$, and $\hrv y \defas \fkt[y](\rv y)$. By \cref{eq:def-cond-mut-inf} we have $\DKL{\kappa_{\hrv x \hrv y | \rv z}(\omega)}{\kappa_{\hrv x | \rv z}(\omega) \times \kappa_{\hrv y | \rv z}(\omega)} = 0$ for $\mu$-\Ae $\omega \in \ps$, which is equivalent to the equality $\mu$-\Ae of the measures $\kappa_{\hrv x \hrv y | \rv z} = \kappa_{\hrv x|\rv z} \times \kappa_{\hrv y|\rv z}$.
  We obtain $\mu$-\Ae,
  \begin{align}
    &\Pcond{\rv x^{-1}(A) \cap \rv y^{-1}(B)}{\rv z} = \kappa_{\hrv x \hrv y | \rv z}(\wc; 1 \times 1) \\
    &\qquad= \kappa_{\hrv x|\rv z}(\wc; 1) \kappa_{\hrv y|\rv z}(\wc; 1) \\
    &\qquad= \Pcond{\rv x^{-1}(A)}{\rv z} \Pcond{\rv y^{-1}(B)}{\rv z} .
  \end{align}

  \noindent
  \labelcref{itm:discrete-cond-ent}:
  See~\cite[Lem.~5.5.6]{Gray1990Entropy}.
  
  \noindent
  \labelcref{itm:kolmogoroff-formula}:
  See~\cite[Lem.~5.5.7]{Gray1990Entropy}.
  
  \noindent
  \labelcref{itm:data-processing}:
  Using \cref{itm:mut-inf-mkv} we have $\condMutInf{\rv x}{\rv z}{\rv y} = 0$ and by \cref{itm:kolmogoroff-formula} it follows that
  \begin{align}
    \mutInf{\rv x}{\rv z} &\le \mutInf{\rv x}{\rv y \rv z} \\
                          &= \mutInf{\rv x}{\rv y} + \condMutInf{\rv x}{\rv z}{\rv y} 
                            = \mutInf{\rv x}{\rv y} . \nonumber\qedhere
  \end{align}%
\end{proof}%
Occasionally we will interpret a probability measure on a finite space $\MMM$ as a vector in $[0,1]^{\MMM}$, equipped with the Borel $\sigma$-algebra. We will use the $L_\infty$-distance on this space.
\begin{definition}
  \label{def:L-infty-dist}
  For two probability measures $\mu$ and $\nu$ on a finite space $\MMM$, their distance is defined as the $L_\infty$-distance $\dist(\mu,\nu) \defas \max_{m \in \MMM} |\mu(m) - \nu(m)|$.
  The diameter of $A \subseteq [0,1]^{\MMM}$ is defined as $\diam(A) = \sup_{\mu, \nu \in A} \dist(\mu, \nu)$.
\end{definition}
\begin{lemma}[{\cite[Lem.~2.7]{Csiszar2011Information}}]
  \label{lem:distance-entropy-bound}
  For two probability measures $\mu$ and $\nu$ on a finite space $\MMM$ with
  $\dist(\mu, \nu) \le \eps \le \frac{1}{2}$ the inequality $|\ent{\mu} - \ent{\nu}| \le - \eps \card{\MMM} \log\eps$ holds.
\end{lemma}

\section{Proof of $\RRRib \subseteq \ol\RRR$}
\label{sec:inner-bound}

For finite spaces $\ps[y]$, $\ps[x]$, and $\ps[u]$, the statement $\RRRib \subseteq \ol\RRR$ is well known, \cf, \cite[Sec.~IV]{Pichler2015Distributed}, \cite[Sec.~III.F]{Courtade2014Multiterminal}. We restate it in the form of the following lemma.
\begin{lemma}
  \label{lem:discrete}
  For random variables $\rv y$, $\rv x$, and $\rv u$ with finite $\ps[y]$, $\ps[x]$, and $\ps[u]$, assume that $\rv y \mkv \rv x \mkv \rv u$ holds.
  Then, for any $\eps > 0$, there exists $n \in \NN$ and a function $f \colon \ps[x]^n \to \MMM$ with $\frac{1}{n} \log\card{\MMM} \le \mutInf[normal]{\rv x}{\rv u}+\eps$ such that
  $\frac{1}{n} \mutInf[big]{\rvt y}{f(\rvt x)} \ge \mutInf[normal]{\rv y}{\rv u} - \eps$.
\end{lemma}

In a first step, we will utilize \cref{lem:discrete} to show $\RRRib \subseteq \ol\RRR$ for an arbitrary alphabet $\ps[x]$,
\ie, we wish to prove the following \cref{pro:achievability}, lifting the restriction $\card{\ps[x]} < \infty$.
\begin{proposition}
  \label{pro:achievability}
  For random variables $\rv y$, $\rv x$, and $\rv u$ with finite $\ps[y]$ and $\ps[u]$, assume that $\rv y \mkv \rv x \mkv \rv u$ holds. Then, for any $\eps > 0$, there exists $n \in \NN$ and a function $f \colon \ps[x]^n \to \MMM$ with $\frac{1}{n} \log\card{\MMM} \le \mutInf[normal]{\rv x}{\rv u}+\eps$ such that
  \begin{align}
    \frac{1}{n} \mutInf[big]{\rvt y}{f(\rvt x)} \ge \mutInf[normal]{\rv y}{\rv u} - \eps . \label{eq:achievability:discreteeq}
  \end{align}
\end{proposition}
\begin{figure}
  \begin{minipage}[b]{.5\linewidth}
    \centering%
    \input{inner-bound.tikz}
    \subcaption{$\RRRib \subseteq \ol\RRR$.}\label{fig:inner-bound}
  \end{minipage}%
  \begin{minipage}[b]{.5\linewidth}
    \centering%
    \input{outer-bound.tikz}%
    \subcaption{$\RRR \subseteq \ol{\RRRib}$.}\label{fig:outer-bound}
  \end{minipage}
  \caption{Illustrations.}\label{fig:illustration}
\end{figure}%
\begin{remark}
  \label{rmk:inner}
  Considering that both definitions of achievability (\cref{def:achievable,def:ib-achievable}) only rely on the notion of mutual information, one may assume that \cref{def:mut-inf} can be used to directly infer \cref{pro:achievability} from \cref{lem:discrete}.
  However, this is not the case.
  For an arbitrary discretization $\fkt[x](\rv x)$ of $\rv x$, we do have $\mutInf{\fkt[x](\rv x)}{\rv u} \le \mutInf{\rv x}{\rv u}$. However, the Markov chain $\rv y \mkv \fkt[x](\rv x) \mkv \rv u$ does not hold in general.
  To circumvent this problem, we will use a discrete random variable $\hrv x = g(\rv x)$ with an appropriate quantizer $g$ and construct a new random variable $\wrv u$, satisfying the Markov chain $\rv y \mkv \hrv x \mkv \wrv u$ such that $\mutInf[normal]{\rv y}{\wrv u}$ is close to $\mutInf{\rv y}{\rv u}$. \Cref{fig:inner-bound} illustrates this strategy. We choose the quantizer $g$ based on the conditional probability distribution of $\rv u$ given $\rv x$, \ie, quantization based on $\kappa_{\rv u|\rv x}$ using $L_\infty$-distance (\cf \cref{def:L-infty-dist}). Subsequently, we will use that, by \cref{lem:distance-entropy-bound}, a small $L_\infty$-distance guarantees a small gap in terms of information measures.
\end{remark}%
\begin{proof}[Proof of \cref{pro:achievability}]
  Let $\mu_{\rv y \rv x \rv u}$ be a probability measure on $\ps \defas \ps[y] \times \ps[x] \times \ps[u]$, such that $\rv y \mkv \rv x \mkv \rv u$ holds. Fix $0 < \delta \le \frac{1}{2}$ and find a finite, measurable partition $(P_i)_{i\in\III}$ of the space of probability measures on $\ps[u]$ such that for every $i \in \III$ we have $\diam(P_i) \le \delta$ and fix some $\nu_i \in P_i$ for every $i \in \III$.
  Define the random variable $\hrv x\colon \ps \to \III$ as $\hrv x = i$ if $\kappa_{\rv u|\rv x} \in P_i$.
  The random variable $\hrv x$ is $\sigma(\rv x)$-measurable (see \cref{sec:x-measurablility}).
  We can therefore find a measurable function $g$ such that $\hrv x = g(\rv x)$ by the factorization lemma \cite[Corollary~1.97]{Klenke2013Probability}.
  Define the new probability space $\ps \times \bigtimes_{i \in \III} \ps[u]$, equipped with the probability measure $\mu_{\rv y \rv x \rv u \wrv u_\III} \defas \mu_{\rv y \rv x \rv u} \times \bigtimes_{i \in \III} \nu_i$. Slightly abusing notation, we define the random variables $\rv y$, $\rv x$, $\rv u$, and $\wrv u_i$ (for every $i \in \III$) as the according projections.
  We also use $\hrv x = g(\rv x)$ and define the random variable $\wrv u = \wrv u_{\hrv x}$. From this construction we have $\mu_{\rv y \rv x \rv u \wrv u_\III}$-\Ae the equality of measures
  $\kappa_{\wrv u|\hrv x} = \kappa_{\wrv u|\rv x} = \nu_{\hrv x}$,
  as well as $\rv y \mkv \hrv x \mkv \wrv u$ and $\rv y \mkv \rv x \mkv \wrv u$%
  \begin{vlong}
    \xspace(see \cref{apx:cond-independence}). 
  \end{vlong}
  \begin{vshort}%
    . This is proven in the extended version \cite{Pichler2018Information}. 
  \end{vshort}
  \xspace Therefore, we have $\mu_{\rv y \rv x \rv u \wrv u_\III}$-\Ae
  \begin{align}
    \dist(\kappa_{\wrv u|\hrv x}, \kappa_{\rv u|\rv x}) &\le \delta ,\text{ and}& \label{eq:proof1:P} 
    \dist(\kappa_{\wrv u|\rv x}, \kappa_{\rv u|\rv x}) \le \delta ,   \end{align}
  by $\kappa_{\wrv u|\hrv x} = \kappa_{\wrv u|\rv x} = \nu_{\hrv x}$ and $\kappa_{\rv u|\rv x}, \nu_{\hrv x} \in P_{\hrv x}$. Thus, for any $u \in \ps[u]$,
  \begin{align}
    \mu_{\rv u}(u) &= \int \kappa_{\rv u|\rv x}(\wc; u) \;d\mu_{\rv y \rv x \rv u} \\
                   &\le \int (\kappa_{\wrv u|\rv x}(\wc; u) + \delta) \;d\mu_{\rv y \rv x \rv u \wrv u_\III} 
                   = \mu_{\wrv u}(u) + \delta
  \end{align}
  and, by the same argument, $\mu_{\rv u}(u) \ge \mu_{\wrv u}(u) - \delta$, \ie, in total,
  \begin{align}
    \label{eq:inner:u-bound}
    \dist(\mu_{\rv u}, \mu_{\wrv u}) \le \delta .
  \end{align}
  Thus, we obtain
  \begin{align}
    \mutInf[normal]{\rv x}{\rv u} &= \ent{\mu_{\rv u}} - \condEnt{\rv u}{\rv x} \label{eq:proof1:discrete-cond-ent} \\
                          &\RL[\cref{eq:inner:u-bound}]\ge \ent{\mu_{\wrv u}} + \delta\card{\ps[u]}\log\delta - \int \ent{\kappa_{\rv u|\rv x}} \;d\mu_{\rv y \rv x \rv u} \label{eq:proof1:dist1} \\
                          &\RL[\cref{eq:proof1:P}]\ge \ent{\mu_{\wrv u}} + 2\delta\card{\ps[u]}\log\delta - \int \ent[big]{\kappa_{\wrv u|\hrv x}} \;d\mu_{\rv y \rv x \rv u \wrv u_\III} \label{eq:proof1:dist2} \\[-1mm]
                          &= \mutInf[normal]{\hrv x}{\wrv u} + 2\delta\card{\ps[u]}\log\delta \label{eq:proof1:discrete-cond-ent2},
  \end{align}
  where \cref{eq:proof1:discrete-cond-ent,eq:proof1:discrete-cond-ent2} follow from \scref{lem:properties}{itm:discrete-cond-ent},
  and in both \cref{eq:proof1:dist1,eq:proof1:dist2} we used \cref{lem:distance-entropy-bound}.
    From $\rv y \mkv \rv x \mkv \rv u$ and \scref{lem:properties}{itm:mut-inf-mkv}, we know that $\mu_{\rv y \rv x \rv u}$-\Ae, we have the equality of mea\-sures $\kappa_{\rv y\rv u|\rv x} = \kappa_{\rv y|\rv x} \times \kappa_{\rv u|\rv x}$.
  Using this equality in \cref{eq:proof1:use-measure-eq} we obtain%
  \begin{align}
    \mu_{\rv y\rv u}(y \times u) &= \int \kappa_{\rv y\rv u|\rv x}(\wc; y \times u) \;d\mu_{\rv y \rv x \rv u}  \label{eq:proof1:cond-prob-apply} \\
                           &= \int \kappa_{\rv y|\rv x}(\wc; y) \kappa_{\rv u|\rv x}(\wc; u) \;d\mu_{\rv y \rv x \rv u} \label{eq:proof1:use-measure-eq} \\
                           &\RL[\cref{eq:proof1:P}]\le \int \kappa_{\rv y|\rv x}(\wc; y) (\kappa_{\wrv u|\rv x}(\wc; u) + \delta) \;d\mu_{\rv y \rv x \rv u \wrv u_\III}  \\
                                                      &\le \int \kappa_{\rv y\wrv u|\rv x}(\wc; y \times u) \;d\mu_{\rv y \rv x \rv u \wrv u_\III} + \delta  \label{eq:proof1:discrete-mkv} \\
                           &= \mu_{\rv y \wrv u}(y \times u) + \delta , \label{eq:proof1:cond-prob-apply2}
  \end{align}
  where \cref{eq:proof1:cond-prob-apply,eq:proof1:cond-prob-apply2} follow from the defining property of conditional probability, \scref{def:cond-expectation}{itm:cond-expectation:tower}, and \cref{eq:proof1:discrete-mkv} follows from $\rv y \mkv \rv x \mkv \wrv u$ and \scref{lem:properties}{itm:mut-inf-mkv}. By the same argument, one can show that $\mu_{\rv y\rv u}(y \times u) \ge \mu_{\rv y \wrv u}(y \times u) - \delta$. Therefore, in total, $\dist(\mu_{\rv y \rv u}, \mu_{\rv y \wrv u}) \le \delta$ and, by \cref{lem:distance-entropy-bound},
  \begin{align}
    \label{eq:proof1:ent-yu-bound}
    |\ent{\rv y \rv u} - \ent[normal]{\rv y \wrv u}| \le -\delta\card{\ps[y]}\card{\ps[u]}\log\delta .
  \end{align}
  Thus, the mutual information can be bounded by
  \begin{align}
    \mutInf[normal]{\rv y}{\rv u} &= \ent{\rv y} + \ent{\rv u} - \ent{\rv y \rv u} \\[-1mm]
                                  &\RL[\cref{eq:inner:u-bound}]\le \ent{\rv y} + \ent[normal]{\wrv u} - \delta \card{\ps[u]} \log\delta  - \ent{\rv y \rv u} \label{eq:proof1:entropy_bound}\\
                                  &\RL[\cref{eq:proof1:ent-yu-bound}]\le \mutInf[normal]{\rv y}{\wrv u} - \delta(\card{\ps[y]}+1)\card{\ps[u]}\log\delta \label{eq:proof1:entropy_bound2} \\
                                  &\le \mutInf[normal]{\rv y}{\wrv u} - 2\delta\card{\ps[y]}\card{\ps[u]}\log\delta , \label{eq:proof1:entropy_bound3}
  \end{align}
  where we applied \cref{lem:distance-entropy-bound} in \cref{eq:proof1:entropy_bound,eq:proof1:entropy_bound2}.
  We apply \cref{lem:discrete} to the three random variables $\rv y$, $\hrv x$, and $\wrv u$ and obtain a function $\hat{f} \colon \III^n \to \MMM$ with $\frac{1}{n} \mutInf[big]{\rvt y}{\hat f(\hrvt x)} \ge \mutInf[normal]{\rv y}{\wrv u} - \delta$ and
  \begin{align}
    \frac{1}{n} \log\card{\MMM} &\le \mutInf[normal]{\hrv x}{\wrv u} + \delta 
                \RL[\cref{eq:proof1:discrete-cond-ent2}]\le \mutInf[normal]{\rv x}{\rv u} + \delta - 2\delta\card{\ps[u]}\log\delta .
  \end{align}
  We have $\hrvt x = g^n \circ \rvt x$ and defining $f \defas \hat{f} \circ g^n$, we obtain
  \begin{align}
    \frac{1}{n} \mutInf[normal]{\rvt y}{f(\rvt x)} &= \frac{1}{n} \mutInf[normal]{\rvt y}{\hat f(\hrvt x)} 
                                                        \ge \mutInf[normal]{\rv y}{\wrv u} - \delta \\[-2mm]
                                                        &\RL[\cref{eq:proof1:entropy_bound3}]\ge \mutInf[normal]{\rv y}{\rv u} +2\delta\card{\ps[y]}\card{\ps[u]}\log\delta - \delta .
  \end{align}
  Choosing $\delta$ such that $\eps \ge -2\delta\card{\ps[y]}\card{\ps[u]}\log\delta + \delta$
  completes the proof.
\end{proof}

We can now complete the proof by showing the following lemma.
\begin{lemma}
  \label{lem:inner}
  $\RRRib \subseteq \ol\RRR$.
\end{lemma}
\begin{proof}
  Assuming $(S,R) \in \RRRib$, choose $\mu_{\rv y \rv x \rv u}$ according to \cref{def:ib-achievable}.
  Clearly $\mutInf{\rv x}{\rv u} < \infty$ to satisfy \cref{eq:Rcondition} and thus also $\mutInf{\rv y}{\rv u} < \infty$ by \scref{lem:properties}{itm:data-processing} as $\rv y \mkv \rv x \mkv \rv u$ holds.
  Pick $\eps > 0$, select functions
  $\fkt[x]$, $\fkt[u]$ such that $\mutInf[big]{\fkt[x](\rv x)}{\fkt[u](\rv u)} \ge \mutInf[normal]{\rv x}{\rv u} - \eps$,
  and select functions $\fktI[y]$, $\fktI[u]$ such that $\mutInf[big]{\fktI[y](\rv y)}{\fktI[u](\rv u)} \ge \mutInf[normal]{\rv y}{\rv u} - \eps$ (\cf \cref{eq:def-cond-mut-inf2}).
  Using $\hrv u \defas \big(\fkt[u](\rv u), \fktI[u](\rv u)\big)$ and $\hrv y \defas \fktI[y](\rv y)$, we have 
  \begin{align}
    0 &= \condMutInf[normal]{\rv y}{\rv u}{\rv x} 
      = \sup_{\fktII[y],\fktII[u]} \condMutInf[normal]{\fktII[y](\rv y)}{\fktII[u](\rv u)}{\rv x} 
      \ge \condMutInf[normal]{\hrv y}{\hrv u}{\rv x} 
      \ge 0 
  \end{align}
  as well as
  \begin{align}
    \mutInf[normal]{\rv x}{\rv u} &= \sup_{\fktII[x], \fktII[u]} \mutInf[big]{\fktII[x](\rv x)}{\fktII[u](\rv u)} \\
                                  &\ge \sup_{\fktII[x]} \mutInf[normal]{\fktII[x](\rv x)}{\hrv u} 
                                  = \mutInf[normal]{\rv x}{\hrv u} ,\text{ and}\label{eq:rate-smaller} \\
    \mutInf[normal]{\rv y}{\rv u} - \eps &\le \mutInf[big]{\fktI[y](\rv y)}{\fktI[u](\rv u)} 
                                           \le \mutInf[normal]{\hrv y}{\hrv u} . \label{eq:mi-similar}
  \end{align}
    We apply \cref{pro:achievability}, substituting $\hrv u \to \rv u$ and $\hrv y \to \rv y$.
  \Cref{pro:achievability} guarantees the existence of a function $f \colon \ps[x]^n \to \MMM$ with $\frac{1}{n} \log\card{\MMM} \le \mutInf[normal]{\rv x}{\hrv u} + \eps \RL[\cref{eq:rate-smaller}]\le \mutInf[normal]{\rv x}{\rv u} + \eps \RL[\cref{eq:Rcondition}]\le R + \eps$ and
  \begin{align}
    \frac{1}{n} \mutInf[normal]{\rvt y}{f(\rvt x)} &= \frac{1}{n} \sup_{\fktII_{\rvt y}} \mutInf[normal]{\fktII_{\rvt y} \circ \rvt y}{f(\rvt x)}\\
                                                   &\ge \frac{1}{n} \mutInf[normal]{\fktI[y]^n \circ \rvt y}{f(\rvt x)} 
                                                   = \frac{1}{n} \mutInf[normal]{\hrvt y}{f(\rvt x)} \\
                                                   &\RL[\cref{eq:achievability:discreteeq}]\ge \mutInf[normal]{\hrv y}{\hrv u} - \eps 
                                                   \RL[\cref{eq:mi-similar}]\ge \mutInf[normal]{\rv y}{\rv u} - 2\eps 
                                                   \RL[\cref{eq:Scondition}]\ge S - 2\eps .
  \end{align}
  Thus, $(S-2\eps, R-\eps) \in \RRR$ and therefore $(S,R) \in \ol\RRR$.
\end{proof}

\section{Proof of $\RRR \subseteq \ol{\RRRib}$}
\label{sec:outer-bound}

We start with the well-known result $\RRRib \subseteq \ol\RRR$ for finite spaces $\ps[y]$, $\ps[x]$, and $\ps[u]$, \cf, \cite[Sec.~IV]{Pichler2015Distributed}, \cite[Sec.~III.F]{Courtade2014Multiterminal}. The statement is rephrased in the following lemma.
\begin{lemma}
  \label{lem:discrete:outer}
  Assume that the spaces $\ps[y]$ and $\ps[x]$ are both finite and $\mu_{\rv y \rv x}$ is fixed. For some $n \in \NN$, let $f \colon \ps[x]^n \to \MMM$ be a function with $\card{\MMM} < \infty$.
  Then there exists a probability measure $\mu_{\rv y \rv x \rv u}$, extending $\mu_{\rv y \rv x}$, such that $\ps[u]$ is finite, $\rv y \mkv \rv x \mkv \rv u$, and
  \begin{align}
    \mutInf{\rv x}{\rv u} &\le \frac{1}{n} \log\card{\MMM} , \label{eq:discrete:R} \\
    \mutInf{\rv y}{\rv u} &\ge \frac{1}{n} \mutInf{\rvt y}{f(\rvt x)} . \label{eq:discrete:S}
  \end{align}
\end{lemma}

We can slightly strengthen \cref{lem:discrete:outer}.
\begin{corollary}
  \label{cor:discrete:outer}
  Assume that, in the setting of \cref{lem:discrete:outer}, we are given $\mu_{\rv z \rv y \rv x}$ on $\ps[z] \times \ps[y] \times \ps[x]$, extending $\mu_{\rv y \rv x}$, where $\ps[z]$ is arbitrary, not necessarily finite. Then there exists a probability measure $\mu_{\rv z \rv y \rv x \rv u}$, extending $\mu_{\rv z \rv y \rv x}$, such that $\ps[u]$ is finite and $\rv z \rv y \mkv \rv x \mkv \rv u$, \cref{eq:discrete:R}, and \cref{eq:discrete:S} hold.
\end{corollary}
\begin{proof}
  Apply \cref{lem:discrete:outer} to obtain $\mu_{\rv y \rv x \rv u}$ on $\ps[y] \times \ps[x] \times \ps[u]$ satisfying \cref{eq:discrete:R}, \cref{eq:discrete:S}, and $\rv y \mkv \rv x \mkv \rv u$.
  We define $\mu_{\rv z \rv y \rv x \rv u}$ by   \begin{align}
    \mu_{\rv z \rv y \rv x \rv u}(A \times y \times x \times u ) &= \frac{\mu_{\rv z \rv y \rv x}(A \times y \times x)}{\mu_{\rv y \rv x}(y \times x)} \mu_{\rv y \rv x \rv u}(y \times x \times u)
  \end{align}
  for any $(y,x,u) \in \ps[y] \times \ps[x] \times \ps[u]$ and $A \in \sa[z]$.
  Pick arbitrary $A \in \sa[z]$, $y \in \ps[y]$, and $u \in \ps[u]$. The Markov chain $\rv z\rv y \mkv \rv x \mkv \rv u$ now follows as the events $\rv z^{-1}(A) \cap \rv y^{-1}(y)$ and $\rv u^{-1}(u)$ are independent given $\rv x^{-1}(x)$ for any $x \in \ps[x]$ (\cf \cref{rmk:conditional-probability}).
\end{proof}

Again, we proceed by extending \cref{cor:discrete:outer}, lifting the restriction that $\ps[x]$ is finite and obtain the following proposition.%
\begin{proposition}
  \label{pro:outer}
  Given a probability measure $\mu_{\rv z \rv y \rv x}$ as in \cref{cor:discrete:outer}, assume that $\card{\ps[y]} < \infty$. For some $n \in \NN$, let $f \colon \ps[x]^n \to \MMM$ be a function with $\card{\MMM} < \infty$.
  Then, for any $\eps > 0$, there exists a probability measure $\mu_{\rv z \rv y \rv x \rv u}$, extending $\mu_{\rv z \rv y \rv x}$ with $\rv z \rv y \mkv \rv x \mkv \rv u$ and 
  \begin{align}
    \mutInf[normal]{\rv x}{\rv u} &\le \frac{1}{n} \log\card{\MMM} \\
    \mutInf[normal]{\rv y}{\rv u} &\ge \frac{1}{n} \mutInf[big]{\rvt y}{f(\rvt x)} - \eps . \label{eq:outer:yu-cond}
  \end{align}
\end{proposition}
\begin{remark}
  \label{rmk:outer_alt}
  In contrast to \cref{pro:achievability}, \Cref{pro:outer} could be proved by the usual single-letterization + time-sharing strategy, by showing that the necessary Markov chains hold.
  However, we will rely on the discrete case (\cref{lem:discrete:outer}) and showcase a technique to lift it to general alphabets.
\end{remark}
\begin{remark}
  \label{rmk:outer}
  In the proof of \cref{pro:outer}, we face a similar problem as outlined in \cref{rmk:inner}.
  We need to construct a function $g(\hrvt x)$ of a ``per-letter'' quantization $\hrvt x \defas \fkt[x]^n(\rvt x)$, that is close to $f(\rvt x)$ in distribution.
  \Cref{fig:outer-bound} provides a sketch.
    \end{remark}
\begin{proof}[Proof of \cref{pro:outer}]
  We can partition $\ps[x]^n = \bigcup_{m \in \MMM} \QQQ_m$ into finitely many measurable, mutually disjoint sets $\QQQ_m \defas f^{-1}(m)$, $m \in \MMM$.
  We want to approximate the sets $\QQQ_m$ by a finite union of rectangles in the semiring \cite[Def. 1.9]{Klenke2013Probability}
  $\Xi \defas \left\{ \BBB : \BBB = \bigtimes_{i=1}^n B_i \text{ with } B_i \in \sa[x] \right\}$.
  We choose $\delta > 0$, which will be specified later. According to \cite[Thm.~1.65(ii)]{Klenke2013Probability}, we obtain $\BBB^{(m)} \defas \bigcup_{k=1}^K \BBB^{(m)}_k$ for each $m \in \MMM$, where $\BBB^{(m)}_k \in \Xi$ are mutually disjoint sets, satisfying
  $\mu_{\rvt x}(\BBB^{(m)} \vartriangle \QQQ_m) \le \delta$.
  Since $\BBB^{(m)}_k \in \Xi$, we have $\BBB^{(m)}_k = \bigtimes_{i=1}^n B^{(m)}_{k,i}$ for some $B^{(m)}_{k,i} \in \sa[x]$.
  We can construct functions $\fkt[x]$ and $g$ such that $g \circ \fkt[x]^n(\vt x) = m$ whenever $\vt x \in \BBB^{(m)}$ and $\vt x \not\in \BBB^{(\not m)}$ with $\BBB^{(\not m)} \defas \bigcup_{m\prm \neq m} \BBB^{(m\prm)}$.
  Indeed, we obtain $\fkt[x]$ by finding a measurable partition of $\ps[x]$ that is finer than $(B^{(m)}_{k,i}, (B^{(m)}_{k,i})\compl)$ for all $i$, $k$, $m$.
  For fixed $m \in \MMM$,   \begin{align}
    &\QQQ_m \subseteq \QQQ_m  \cup \big( \BBB^{(m)}\setminus\BBB^{(\not m)} \big)  \\
                      &\;\subseteq \big( \BBB^{(m)}\setminus\BBB^{(\not m)} \big) \cup  \big( \QQQ_m \setminus \BBB^{(m)} \big) \cup \bigcup_{m\prm \neq m} \QQQ_m \cap \BBB^{(m\prm)} \\
           &\;\subseteq \big( \BBB^{(m)}\setminus\BBB^{(\not m)} \big) \cup  \big( \QQQ_m \vartriangle \BBB^{(m)} \big) \cup \bigcup_{m\prm \neq m} \BBB^{(m\prm)} \setminus \QQQ_{m\prm} \label{eq:proof-outer:disjoint}\\
           &\;\subseteq \big( \BBB^{(m)}\setminus\BBB^{(\not m)} \big) \cup \bigcup_{m\prm} \BBB^{(m\prm)} \vartriangle \QQQ_{m\prm} , \label{eq:proof-outer:disjoint2}
  \end{align}
  where we used the fact that $\QQQ_m \cap \QQQ_{m\prm} = \nil$ for $m \neq m\prm$ in \cref{eq:proof-outer:disjoint}.
  Using $\hrv x \defas \fkt[x](\rv x)$, we obtain for any $\vt y \in \ps[y]^n$
  \begin{align}
    &\mu_{\rvt y f(\rvt x)}(\vt y \times m) = \mu_{\rvt y \rvt x}(\vt y \times \QQQ_m) \\
    &\qquad\RL[\cref{eq:proof-outer:disjoint2}]\le \mu_{\rvt y \rvt x}\big(\vt y \times (\BBB^{(m)}\setminus\BBB^{(\not m)})\big) + \sum_{m\prm} \mu_{\rvt x}(\BBB^{(m\prm)} \vartriangle \QQQ_{m\prm}) \\
    &\qquad\le \mu_{\rvt y g(\hrvt x)}(\vt y \times m) + \card{\MMM}\delta . \label{eq:proof-outer:prob-close1}
  \end{align}
  On the other hand, we have
  \begin{align}
    &\mu_{\rvt y f(\rvt x)}(\vt y \times m) = \mu_{\rvt y}(\vt y) - \sum_{m\prm \neq m} \mu_{\rvt y f(\rvt x)}(\vt y \times m\prm) \\
    &\qquad\RL[\cref{eq:proof-outer:prob-close1}]\ge \mu_{\rvt y}(\vt y) - \sum_{m\prm \neq m} \big(\mu_{\rvt y g(\hrvt x)}(\vt y \times m\prm) + \card{\MMM}\delta \big) \\
    &\qquad\ge \mu_{\rvt y g(\hrvt x)}(\vt y \times m) - \card{\MMM}^2\delta .
  \end{align}
  We thus obtain $\dist(\mu_{\rvt y f(\rvt x)}, \mu_{\rvt y g(\hrvt x)}) \le \card{\MMM}^2\delta$.
  This also implies $\dist(\mu_{f(\rvt x)}, \mu_{g(\hrvt x)}) \le \card{\ps[y]}^n\card{\MMM}^2\delta$.
  Assume $\card{\ps[y]}^n\card{\MMM}^2\delta \le \frac{1}{2}$ and apply \cref{cor:discrete:outer} substituting $\hrv x \to \rv x$, $\rv x \rv z \to \rv z$, and the function $g \to f$.
  This yields a random variable $\rv u$ with $\rv x \rv z \rv y \mkv \hrv x \mkv \rv u$,
  \begin{align}
    \mutInf[normal]{\hrv x}{\rv u} &\le \frac{1}{n} \log\card{\MMM} , \text{ and} &
    \mutInf{\rv y}{\rv u} &\ge \frac{1}{n} \mutInf{\rvt y}{g(\hrvt x)} . \label{eq:proof-outer:xuR-yu}
  \end{align}
  We also obtain $\rv z \rv y \mkv \rv x \mkv \rv u$ due to
  \begin{align}
    0 &= \condMutInf[normal]{\rv x \rv z \rv y}{\rv u}{\hrv x} \label{eq:proof2:use-mkv} \\
      &= \mutInf[normal]{\rv x \rv z \rv y}{\rv u} - \mutInf[normal]{\rv u}{\hrv x} \label{eq:proof2:mutinf1} \\
      &\ge \mutInf[normal]{\rv x \rv z \rv y}{\rv u} - \mutInf[normal]{\rv u}{\rv x} \label{eq:proof2:data-proc} \\
      &= \condMutInf[normal]{\rv z\rv y}{\rv u}{\rv x} \label{eq:proof2:mutinf2} \\
      &\ge 0 , \label{eq:proof2:nonneg}
  \end{align}
  where \cref{eq:proof2:use-mkv} follows from $\rv x \rv z \rv y  \mkv \hrv x \mkv \rv u$ using \scref{lem:properties}{itm:mut-inf-mkv}, \cref{eq:proof2:mutinf1,eq:proof2:mutinf2} follow from \scref{lem:properties}{itm:kolmogoroff-formula}, \cref{eq:proof2:data-proc} is a consequence of \cref{def:mut-inf}, and we used \scref{lem:properties}{itm:mut-inf-mkv} in \cref{eq:proof2:nonneg}.
  This also immediately implies $0=\condMutInf[normal]{\rv x}{\rv u}{\hrv x}$ and hence
  \begin{align}
    \frac{1}{n} \log\card{\MMM} &\RL[\cref{eq:proof-outer:xuR-yu}]\ge \mutInf[normal]{\hrv x}{\rv u} 
                                = \mutInf[normal]{\hrv x}{\rv u} + \condMutInf[normal]{\rv x}{\rv u}{\hrv x} \\
                                &= \mutInf[normal]{\rv x \hrv x}{\rv u}  \label{eq:proof2:use-kolmogoroff}
                                = \mutInf[normal]{\rv x}{\rv u} ,
  \end{align}
  where we used \scref{lem:properties}{itm:kolmogoroff-formula} in \cref{eq:proof2:use-kolmogoroff}.
  We also have
  \begin{align}
    \mutInf[normal]{\rv y}{\rv u} &\RL[\cref{eq:proof-outer:xuR-yu}]\ge \frac{1}{n} \mutInf{\rvt y}{g(\hrvt x)} \\
                                  &= \frac{1}{n} \big( \ent[normal]{\rvt y} + \ent[normal]{g(\hrvt x)} - \ent[normal]{\rvt y g(\hrvt x)} \big) \\
                                  &\ge \frac{1}{n} \mutInf[big]{\rvt y}{f(\rvt x)} + \frac{1}{n} \card{\ps[y]}^{n} \card{\MMM}^3\delta \log(\card{\MMM}^2\delta)
                                    \eqnl &\qquad+ \frac{1}{n} \card{\ps[y]}^n \card{\MMM}^3\delta \log(\card{\ps[y]}^n\card{\MMM}^2\delta) \label{eq:proof2:use-ent-dist-bound} \\
                                  &\ge \frac{1}{n} \mutInf[big]{\rvt y}{f(\rvt x)} + \frac{2}{n} \card{\ps[y]}^{n}\card{\MMM}^3\delta \log(\card{\MMM}^2\delta) 
  \end{align}
  where we used \cref{lem:distance-entropy-bound} in \cref{eq:proof2:use-ent-dist-bound}.
  Select $\delta$ such that $\eps \ge -\frac{2}{n} \card{\ps[y]}^{n}\card{\MMM}^3\delta \log(\card{\MMM}^2\delta)$.%
\end{proof}

We can now finish the proof by showing the following lemma.%
\begin{lemma}
  \label{lem:outer}
  $\RRR \subseteq \ol{\RRRib}$.
\end{lemma}
\begin{proof}
  Assume $(S,R) \in \RRR$ and choose $n \in \NN$ and $f$, satisfying $\frac{1}{n} \log\card{\MMM} \le R$ and \cref{eq:achievable}. Choose any $\eps > 0$ and find $\fkt[y]$ such that%
  \begin{align}
    \label{eq:outer-proof:1}
    \mutInf[big]{\fkt[y]^n(\rvt y)}{f(\rvt x)} \ge \mutInf[big]{\rvt y}{f(\rvt x)} - \eps \RL[\cref{eq:achievable}]\ge nS - \eps .
  \end{align}
  This is possible by applying \cite[Lem. 5.2.2]{Gray1990Entropy} with the algebra that is generated by the rectangles (\cf the paragraph above \cite[Lem.~5.5.1]{Gray1990Entropy}).
    We apply \cref{pro:outer}, substituting $\fkt[y](\rv y) \to \rv y$ and $\rv y \to \rv z$.
  For arbitrary $\eps > 0$, \cref{pro:outer} provides $\rv u$ with $\rv y \fkt[y](\rv y)  \mkv \rv x \mkv \rv u$ (\ie, $\rv y \mkv \rv x \mkv \rv u$) and
  \begin{align}
    \mutInf[normal]{\rv x}{\rv u} &\le \frac{1}{n} \log\card{\MMM} \le R \\
    \mutInf[normal]{\rv y}{\rv u} &\ge \mutInf[normal]{\fkt[y](\rv y)}{\rv u} \\
                                  &\RL[\cref{eq:outer:yu-cond}]\ge \frac{1}{n} \mutInf[big]{\fkt[y]^n(\rvt y)}{f(\rvt x)} - \eps 
                                  \RL[\cref{eq:outer-proof:1}] \ge S - 2\eps .
  \end{align}
  Hence, $(S - 2\eps, R) \in \RRRib$ and consequently $(S,R) \in \ol\RRRib$.
\end{proof}

\appendix

\section{Proof}
\label{sec:proof}

\subsection{$\hrv x$ is $\sigma(\rv x)$-measurable}
\label{sec:x-measurablility}

For $u \in \ps[u]$ consider the $\sigma(\rv x)$-measurable function $h_u \defas \kappa_{\rv u|\rv x}(\wc; u)$ on $[0,1]$. We obtain the vector valued function $h \defas (h_u)_{u \in \ps[u]}$ on $[0,1]^{\card{\ps[u]}}$. This function $h$ is $\sigma(\rv x)$-measurable as every component is $\sigma(\rv x)$-measurable.
Thus, we have $\hrv x^{-1}(i) = h^{-1}(P_i) \in \sigma(\rv x)$.

\begin{vlong}
\subsection{Distribution of $\wrv u$ and Conditional Independence}
\label{apx:cond-independence}

We will first show that $\mu_{\rv y \rv x \rv u \wrv u_\III}$-\Ae
\begin{align}
  \label{eq:kappa-is-nu}
  \kappa_{\wrv u|\hrv x} = \kappa_{\wrv u|\rv x} = \nu_{\hrv x} .
\end{align}
Clearly, $\nu_{\hrv x}$ is a probability measure everywhere. Fixing $u \in \ps[u]$, we need that $\nu_{\hrv x}(u)$ is $\sigma(\hrv x)$-measurable, which is shown by the factorization lemma \cite[Corollary~1.97]{Klenke2013Probability}, when writing $\nu_{\hrv x}(u) = \nu_{(\wc)}(u) \circ \hrv x$. Also, this proves $\sigma(\rv x)$-measurability as $\hrv x$ is $\sigma(\rv x)$-measurable, \ie, $\sigma(\hrv x) \subseteq \sigma(\rv x)$.
It remains to show the defining property of conditional probability, \scref{def:cond-expectation}{itm:cond-expectation:tower}. Choosing $B \in \sigma(\rv x)$ and $u \in \ps[u]$, we need to show that
\begin{align}
  \Exp{\ind{B}{} \nu_{\hrv x}(u)} = \Exp{\ind{B}{} \ind{\{u\}}{\wrv u}} . \label{eq:proof-independence-exp-equal}
\end{align}
The statement for $B \in \sigma(\hrv x)$ then follows by $\sigma(\hrv x) \subseteq \sigma(\rv x)$, \ie, the $\sigma(\rv x)$-measurability of $\hrv x$.
We prove \cref{eq:proof-independence-exp-equal} by
\begin{align}
  \Exp{\ind{B}{} \nu_{\hrv x}(u)} &= \sum_{i \in \III} \Exp{\ind{i}{\hrv x} \ind{B}{} \nu_i(u)} \\
                                  &= \sum_{i \in \III} \nu_i(u) \Exp{\ind{i}{\hrv x} \ind{B}{} } \\
                                  &= \sum_{i \in \III} \Exp{\ind{u}{\wrv u_i}} \Exp{\ind{i}{\hrv x} \ind{B}{} } \\
                                  &= \sum_{i \in \III} \Exp{\ind{i}{\hrv x} \ind{B}{} \ind{u}{\wrv u_i}} \label{eq:markov-proof-fubini} \\
                                  &= \sum_{i \in \III} \Exp{\ind{i}{\hrv x} \ind{B}{} \ind{u}{\wrv u}} \\
                                  &= \Exp{\ind{B}{} \ind{u}{\wrv u}} , \label{eq:markov-proof-eq1}
\end{align}
where we used Fubini's theorem \cite[Thm.~14.16]{Klenke2013Probability} in \cref{eq:markov-proof-fubini}.

To prove $\condMutInf[normal]{\rv y}{\wrv u}{\rv x} = 0$, we need to show that for every $y \in \ps[y]$, $u \in \ps[u]$, and $B \in \sigma(\rv x)$, we have
\begin{align}
  \int \ind{B}{} \kappa_{\rv y|\rv x}(\wc; y) \nu_{\hrv x}(u) \;d\mu_{\rv y \rv x \rv u} &=
  \int \ind{B}{} \ind{u}{\wrv u} \ind{y}{\rv y} \;d\mu_{\rv y \rv x \rv u \wrv u_\III}
\end{align}
and by integrating, we indeed obtain
\begin{align}
  &\int \ind{B}{} \kappa_{\rv y|\rv x}(\wc; y) \nu_{\hrv x}(u) \;d\mu_{\rv y \rv x \rv u} \\
  &= \sum_{i \in \III} \int \ind{B}{}\ind{i}{\hrv x} \kappa_{\rv y|\rv x}(\wc; y) \nu_i(u) \;d\mu_{\rv y \rv x \rv u} \\
  &= \sum_{i \in \III} \nu_i(u) \int \ind{B}{}\ind{i}{\hrv x} \kappa_{\rv y|\rv x}(\wc; y)  \;d\mu_{\rv y \rv x \rv u} \\
  &= \sum_{i \in \III} \int \ind{u}{\wrv u_i} \;d\mu_{\wrv u_\III} \int \ind{B}{}\ind{i}{\hrv x} \ind{y}{\rv y}  \;d\mu_{\rv y \rv x \rv u} \label{eq:markov-proof:cond-prob} \\
  &= \sum_{i \in \III} \int \ind{B}{} \ind{u}{\wrv u_i} \ind{i}{\hrv x} \ind{y}{\rv y}  \;d\mu_{\rv y \rv x \rv u \wrv u_\III} \label{eq:markov-proof-fubini2}\\
  &= \sum_{i \in \III} \int \ind{B}{} \ind{u}{\wrv u} \ind{i}{\hrv x} \ind{y}{\rv y}  \;d\mu_{\rv y \rv x \rv u \wrv u_\III}  \\
  &= \int \ind{B}{} \ind{u}{\wrv u} \ind{y}{\rv y}  \;d\mu_{\rv y \rv x \rv u \wrv u_\III} ,
  \end{align}
where we used \scref{def:cond-expectation}{itm:cond-expectation:tower} in \cref{eq:markov-proof:cond-prob} and Fubini's theorem \cite[Thm.~14.16]{Klenke2013Probability} in \cref{eq:markov-proof-fubini2}. By replacing $\kappa_{\rv y|\rv x}$ with $\kappa_{\rv y|\hrv x}$ and using $B \in \sigma(\hrv x)$, the same argument can be used to show $\condMutInf[normal]{\rv y}{\wrv u}{\hrv x} = 0$.
\end{vlong}

\section*{Acknowledgment}
The authors would like to thank Michael Meidlinger for providing inspiration for this work.

\bibliographystyle{myIEEEtran}
\bibliography{IEEEabrv,literature}

\end{document}

%% file: packages.tex
\usepackage[utf8]{inputenc}
\usepackage{import}

\ifdefined\uselmodern
\usepackage{lmodern}
\fi

\usepackage[T1]{fontenc}
\usepackage{etex}
\usepackage{color}
\usepackage{microtype}
\usepackage{etoolbox}\usepackage{ifthen}
\usepackage{xparse}
\usepackage{substr}
\usepackage{subcaption}
\usepackage[hide=true,dir=bashful]{bashful}
\usepackage{calculator}
\usepackage{changepage}
\usepackage{xspace}

\makeatletter
  \@ifclassloaded{IEEEtran}{
      }{
    \usepackage{biblatex}
    \addbibresource{IEEEabrv.bib}
    \addbibresource{literature.bib}
  }
\makeatother

\usepackage{mathtools}

\usepackage{booktabs}
\usepackage{array}

\usepackage{dsfont}
\usepackage{mathrsfs}

\usepackage{setspace}

\usepackage[justification=centering]{caption}

\usepackage{colonequals}
\usepackage{centernot}
\usepackage{relsize}
\usepackage{tikz}
\usepackage{pgfplots}
\usepackage{longtable}
\usepackage{eqparbox}
\usepackage{fp}

\usepackage{comment}

\makeatletter
\@ifclassloaded{beamer}{
  \usepackage{appendixnumberbeamer}
}{
  \usepackage{amssymb}
  \usepackage{environ}
  \usepackage{amsmath}
  \usepackage[bottom]{footmisc}
  \usepackage[thepage]{zref}
  \usepackage{enumitem}
  \@ifclassloaded{IEEEtran}{}{
    \@ifclassloaded{a0poster}{}{
      \@ifpackageloaded{classicthesis}{}{
        \usepackage[heightrounded=true]{geometry}
        }
    }
    \usepackage[nottoc]{tocbibind}
    \usepackage{imakeidx}
  }
  \@ifundefined{nolinks}{
    \usepackage{hyperref}
    \@ifclassloaded{a0poster}{}{
      \usepackage[all]{hypcap}
    }
    \usepackage{bookmark}
  }{
    
  }
}
\makeatother

\usepackage{amsthm}
\usepackage{thmtools, thm-restate}
\usepackage[capitalise]{cleveref}

\ifdefined\useautonum
\makeatletter
\@ifclassloaded{beamer}{
  \usepackage{scrpage2}
}{
  \usepackage{scrlayer-scrpage}
  \usepackage{autonum}
  \ifdef{\Cref}{%
    \autonum@generatePatchedReferenceCSL{Cref}%
  }{}%
}
\makeatother
\fi

\usepackage{scalerel}
\usepackage{stackengine,wasysym}

%% file: custom_environments.tex
\crefname{apx}{appendix}{appendices}
\Crefname{apx}{Appendix}{Appendices}
\crefformat{apx}{the #2appendix#3}
\Crefformat{apx}{the #2Appendix#3}

\newtheorem*{conjecture*}{Conjecture}
\newtheorem*{definition*}{Definition}
\newtheorem*{proposition*}{Proposition}
\newtheorem*{corollary*}{Corollary}
\makeatletter
\@ifclassloaded{beamer}{
  
  \newtheorem*{lemma*}{Lemma}
  \newtheorem*{theorem*}{Theorem}
}{
  \ifdef{\numbertheoremwithin}{
    \newtheorem{theorem}{Theorem}[\numbertheoremwithin]
  }{
    \newtheorem{theorem}{Theorem}
  }
  \newtheorem*{theorem*}{Theorem}

  \newcommand{\mynewtheorem}[2]{
    \ifdef{\seperatenumbering}{
      \newtheorem{#1}{#2}
    }{
      \newtheorem{#1}[theorem]{#2}
    }
  }
  \mynewtheorem{lemma}{Lemma}
  \mynewtheorem{corollary}{Corollary}
  \mynewtheorem{definition}{Definition}
  \mynewtheorem{proposition}{Proposition}
  \mynewtheorem{conjecture}{Conjecture}
}
\makeatother

\theoremstyle{remark}
\makeatletter
\@ifclassloaded{beamer}{}{
  
}
\makeatother

\newtheorem{remark}{Remark}

\newtheorem*{remark*}{Remark}
\newtheorem*{example*}{Example}

\crefname{conjecture}{Conjecture}{Conjectures}
\Crefname{conjecture}{Conjecture}{Conjectures}

\crefname{enumi}{part}{parts}
\Crefname{enumi}{Part}{Parts}

\crefname{equation}{Equation}{Equations}
\Crefname{equation}{Equation}{Equations}
\crefformat{equation}{#2(#1)#3}
\crefrangeformat{equation}{#3(#1)#4--#5(#2)#6}
\crefmultiformat{equation}{#2(#1)#3}{ and~#2(#1)#3}{, #2(#1)#3}{ and~#2(#1)#3}
\Crefformat{equation}{#2(#1)#3}
\Crefrangeformat{equation}{#3(#1)#4--#5(#2)#6}
\Crefmultiformat{equation}{#2(#1)#3}{ and~#2(#1)#3}{, #2(#1)#3}{ and~#2(#1)#3}
\crefname{figure}{Figure}{Figures}
\Crefname{figure}{Figure}{Figures}

\crefname{page}{page}{page}
\Crefname{page}{Page}{Page}

\makeatletter
\@ifclassloaded{IEEEtran}{}{
  \crefname{figure}{Figure}{Figures}
  \Crefname{figure}{Figure}{Figures}
}
\makeatother

\crefname{assumption}{assumption}{assumptions}
\Crefname{assumption}{Assumption}{Assumptions}

\crefname{case}{case}{cases}
\Crefname{case}{Case}{Cases}

%% file: custom_commands.tex
\let\originalleft\left
\let\originalright\right
\let\originalmiddle\middle
\renewcommand{\left}{\mathopen{}\mathclose\bgroup\originalleft}
\renewcommand{\right}{\aftergroup\egroup\originalright}
\renewcommand{\middle}{\originalmiddle}

\newcommand{\scref}[2]{{\cref{#2} of~\cref{#1}}}

\makeatletter
\@ifclassloaded{beamer}{%
  \newcommand{\todo}[1][\empty]{%
    \textbf{\textcolor{red}{\ifthenelse{\equal{#1}{\empty}}{TODO}{TODO:~\emph{#1}}}}%
  }%
}
{
  \newcounter{todo}
  \newcommand{\todo}[1][\empty]{%
    \bookmarksetupnext{keeplevel}    \subpdfbookmark{TODO}{todo.\arabic{todo}}%
    \bookmarksetup{keeplevel=false}%
    \refstepcounter{todo}%
    \textcolor{red}{\textbf{TODO}%
      \ifthenelse{\equal{#1}{\empty}}{}%
      {\footnote{\textcolor{red}{\emph{#1}}}}%
    }%
    \xspace%
  }
}
\makeatother

\makeatletter
\@ifclassloaded{beamer}{
  \newcommand\changemode[1]{%
    \gdef\beamer@currentmode{#1}}
}{}
\makeatother

\makeatletter
\@ifclassloaded{IEEEtran}{
  
}{}
\makeatother

\newcounter{align}
\expandafter\let\expandafter\oldalign\csname align\endcsname
\expandafter\def\csname align\endcsname{\refstepcounter{align}\oldalign}

\newcommand{\RL}[2][\empty]{\stackrel{\mathclap{\text{\scriptsize#1}}}{#2}}

\newcommand{\ie}{i.\,e.\@\xspace}
\newcommand{\Ae}{a.\,e.\@\xspace}

\newcommand{\eg}{e.\,g.\@\xspace}

\newcommand{\cf}{cf.\@\xspace}

\newcommand{\iid}{i.i.d.\@\xspace}

\newcommand{\wrt}{w.r.t.\@\xspace}

\include{custom_git}

\newcommand{\eqnl}{\nonumber\\*}

\newcommand{\prm}{^\prime}

%% file: custom_math.tex
\DeclareDocumentCommand{\card}{ O{\empty} m}{
  \ensuremath{
    \ifthenelse{\equal{#1}{big}}{\big\vert #2 \big\vert}{
      \ifthenelse{\equal{#1}{Big}}{\Big\vert #2 \Big\vert}{
        \ifthenelse{\equal{#1}{bigg}}{\bigg\vert #2 \bigg\vert}{
          \ifthenelse{\equal{#1}{Bigg}}{\Bigg\vert #2 \Bigg\vert}{
            \ifthenelse{\equal{#1}{normal}}{\vert #2 \vert}{\left\vert #2 \right\vert}
          }
        }
      }
    }
  }
}

\DeclareDocumentCommand{\cvx}{ O{\empty} m}{
  \ensuremath{
    \ensuremath{\mathrm{conv}
      \ifthenelse{\equal{#1}{big}}{\big(#2\big)}{
        \ifthenelse{\equal{#1}{Big}}{\Big(#2\Big)}{
          \ifthenelse{\equal{#1}{bigg}}{\bigg(#2\bigg)}{
            \ifthenelse{\equal{#1}{Bigg}}{\Bigg(#2\Bigg)}{
              \ifthenelse{\equal{#1}{normal}}{(#2)}{\left(#2\right)}
            }
          }
        }
      }
    }
  }
}

\newcommand{\ol}[1]{\ensuremath{\overline{#1}}}

\DeclareDocumentCommand{\Ntoo}{ O{\empty} m}{
  \ensuremath{
    \ifthenelse{\equal{#1}{big}}{\big[#2\big]}{
      \ifthenelse{\equal{#1}{Big}}{\Big[#2\Big]}{
        \ifthenelse{\equal{#1}{bigg}}{\bigg[#2\bigg]}{
          \ifthenelse{\equal{#1}{Bigg}}{\Bigg[#2\Bigg]}{
            \ifthenelse{\equal{#1}{normal}}{[#2]}{\left[#2\right]}
          }
        }
      }
    }
  }
}
\DeclareDocumentCommand{\NZtoo}{ O{\empty} m}{
  \ensuremath{
    \ifthenelse{\equal{#1}{big}}{\big[0\,{:}\,#2\big]}{
      \ifthenelse{\equal{#1}{Big}}{\Big[0\,{:}\,#2\Big]}{
        \ifthenelse{\equal{#1}{bigg}}{\bigg[0\,{:}\,#2\bigg]}{
          \ifthenelse{\equal{#1}{Bigg}}{\Bigg[0\,{:}\,#2\Bigg]}{
            \ifthenelse{\equal{#1}{normal}}{[0\,{:}\,#2]}{\left[0\,{:}\,#2\right]}
          }
        }
      }
    }
  }
}

\newcommand{\defas}{\vcentcolon=}

\newcommand{\wt}[1]{\ensuremath{\widetilde{#1}}}

\newcommand{\compl}{\ensuremath{^{\mathrm{c}}}}

\newcommand{\NN}{\ensuremath{\mathbb{N}}}

\newcommand{\RR}{\ensuremath{\mathbb{R}}}

\newcommand{\vt}[1]{\ensuremath{\boldsymbol{#1}}}

\newcommand{\wc}{\ensuremath{{}\cdot{}}}

\DeclareDocumentCommand{\ind}{m m}{%
  \ensuremath{%
    \mathds{1}_{#1}%
    \ifthenelse{\equal{#2}{\empty}}{}{(#2)}%
  }%
}

%% file: custom_probability.tex
\DeclareDocumentCommand{\Exp}{ O{\empty} O{\empty} m}{
  \ensuremath{\mathds{E}_{#1}
    \ifthenelse{\equal{#2}{big}}{\big[ #3 \big]}{
      \ifthenelse{\equal{#2}{Big}}{\Big[ #3 \Big]}{
        \ifthenelse{\equal{#2}{bigg}}{\bigg[ #3 \bigg]}{
          \ifthenelse{\equal{#2}{Bigg}}{\Bigg[ #3 \Bigg]}{
            \ifthenelse{\equal{#2}{normal}}{[ #3 ]}{\left[ #3 \right]}
          }
        }
      }
    }
  }
}

\DeclareDocumentCommand{\condExp}{ O{\empty} m m}{
  \ensuremath{\mathds{E}
    \ifthenelse{\equal{#1}{big}}{\big[{#2\big| #3}\big]}{
      \ifthenelse{\equal{#1}{Big}}{\Big[{#2\Big| #3}\Big]}{
        \ifthenelse{\equal{#1}{bigg}}{\bigg[{#2\bigg| #3}\bigg]}{
          \ifthenelse{\equal{#1}{Bigg}}{\Bigg[{#2\Bigg| #3}\Bigg]}{
            \ifthenelse{\equal{#1}{normal}}{[{#2| #3}]}{\left[#2 \middle| #3\right]}
          }
        }
      }
    }
  }
}

\newcommand{\p}[2][\empty]{\ensuremath{\mathrm{p}_{#1}\ifthenelse{\equal{#2}{}}{}{\left({#2}\right)}}}
\newcommand{\q}[2][\empty]{\ensuremath{\mathrm{q}_{#1}\ifthenelse{\equal{#2}{}}{}{\left({#2}\right)}}}
\newcommand{\wtp}[2][\empty]{\ensuremath{\wt{\mathrm{p}}_{#1}\ifthenelse{\equal{#2}{}}{}{\left({#2}\right)}}}
\newcommand{\htp}[2][\empty]{\ensuremath{\hat{\mathrm{p}}_{#1}\ifthenelse{\equal{#2}{}}{}{\left({#2}\right)}}}
\newcommand{\pp}[2][\empty]{\ensuremath{\mathrm{p}'_{#1}\ifthenelse{\equal{#2}{}}{}{\left({#2}\right)}}}
\newcommand{\ppp}[2][\empty]{\ensuremath{\mathrm{p}''_{#1}\ifthenelse{\equal{#2}{}}{}{\left({#2}\right)}}}

\DeclareDocumentCommand{\Prob}{ O{\empty} O{\empty} m}{
  \ensuremath{\mathrm{P}\ifthenelse{\equal{#1}{\empty}}{}{_{#1}}
    \ifthenelse{\equal{#2}{big}}{\big\{#3\big\}}{
      \ifthenelse{\equal{#2}{Big}}{\Big\{#3\Big\}}{
        \ifthenelse{\equal{#2}{bigg}}{\bigg\{#3\bigg\}}{
          \ifthenelse{\equal{#2}{Bigg}}{\Bigg\{#3\Bigg\}}{
            \ifthenelse{\equal{#2}{normal}}{\{#3\}}{\left\{#3\right\}}
          }
        }
      }
    }
  }
}

\DeclareDocumentCommand{\DKL}{ O{\empty} m m}{
  \ensuremath{\mathrm{D}
    \ifthenelse{\equal{#1}{big}}{\big(#2 \big\Vert #3\big)}{
      \ifthenelse{\equal{#1}{Big}}{\Big(#2 \Big\Vert #3\Big)}{
        \ifthenelse{\equal{#1}{bigg}}{\bigg(#2 \bigg\Vert #3\bigg)}{
          \ifthenelse{\equal{#1}{Bigg}}{\Bigg(#2 \Bigg\Vert #3\Bigg)}{
            \ifthenelse{\equal{#1}{normal}}{(#2 \Vert #3)}{\left(#2 \middle\Vert #3\right)}
          }
        }
      }
    }
  }
}

\DeclareDocumentCommand{\pcond}{ O{\empty} O{\empty} O{\empty} m m}{
  \ensuremath{\mathrm{p}
    \ifthenelse{\equal{#1#2}{\empty}}{}{_{#1|#2}}
    \ifthenelse{\equal{#4#5}{\empty}}{}{
      \ifthenelse{\equal{#3}{big}}{\big(#4 \big| #5\big)}{
        \ifthenelse{\equal{#3}{Big}}{\Big(#4 \Big| #5\Big)}{
          \ifthenelse{\equal{#3}{bigg}}{\bigg(#4 \bigg| #5\bigg)}{
            \ifthenelse{\equal{#3}{Bigg}}{\Bigg(#4 \Bigg| #5\Bigg)}{
              \ifthenelse{\equal{#3}{normal}}{(#4 | #5)}{\left(#4 \middle| #5\right)}
            }
          }
        }
      }
    }
  }
}

\DeclareDocumentCommand{\Pcond}{ O{\empty} O{\empty} O{\empty} m m}{
  \ensuremath{\mathrm{P}
    \ifthenelse{\equal{#1#2}{\empty}}{}{_{#1|#2}}
    \ifthenelse{\equal{#4#5}{\empty}}{}{
      \ifthenelse{\equal{#3}{big}}{\big\{#4 \big| #5\big\}}{
        \ifthenelse{\equal{#3}{Big}}{\Big\{#4 \Big| #5\Big\}}{
          \ifthenelse{\equal{#3}{bigg}}{\bigg\{#4 \bigg| #5\bigg\}}{
            \ifthenelse{\equal{#3}{Bigg}}{\Bigg\{#4 \Bigg| #5\Bigg\}}{
              \ifthenelse{\equal{#3}{normal}}{\{#4 | #5\}}{\left\{#4 \middle| #5\right\}}
            }
          }
        }
      }
    }
  }
}

\newcommand{\rv}[1]{\ensuremath{\mathsf{\uppercase{#1}}}}
\newcommand{\rvt}[1]{\vt{\rv{#1}}} 
\newcommand{\wrv}[1]{\wt{\rv{#1}}}  
\newcommand{\hrv}[1]{\hat{\rv{#1}}} \newcommand{\hrvt}[1]{\hat{\rvt{#1}}} 

\DeclareDocumentCommand{\mutInf}{ O{\empty} m m}{
  \ensuremath{\mathrm{I}
    \ifthenelse{\equal{#1}{big}}{\big({#2; #3}\big)}{
      \ifthenelse{\equal{#1}{Big}}{\Big({#2; #3}\Big)}{
        \ifthenelse{\equal{#1}{bigg}}{\bigg({#2; #3}\bigg)}{
          \ifthenelse{\equal{#1}{Bigg}}{\Bigg({#2; #3}\Bigg)}{
            \ifthenelse{\equal{#1}{normal}}{({#2; #3})}{\left({#2; #3}\right)}
          }
        }
      }
    }
  }
}

\DeclareDocumentCommand{\condMutInf}{ O{\empty} m m m }{
  \ensuremath{\mathrm{I}
    \ifthenelse{\equal{#1}{big}}{\big(#2; #3 \big| #4\big)}{
      \ifthenelse{\equal{#1}{Big}}{\Big(#2; #3 \Big| #4\Big)}{
        \ifthenelse{\equal{#1}{bigg}}{\bigg(#2; #3 \bigg| #4\bigg)}{
          \ifthenelse{\equal{#1}{Bigg}}{\Bigg(#2; #3 \Bigg| #4\Bigg)}{
            \ifthenelse{\equal{#1}{normal}}{(#2; #3 | #4)}{\left(#2; #3 \middle| #4\right)}
          }
        }
      }
    }
  }
}

\DeclareDocumentCommand{\typ}{ O{\empty} O{\empty} m}{
  \ensuremath{\mathcal{T}^{#1}_{[#3]#2}}
}
\DeclareDocumentCommand{\type}{ O{\empty} m}{
  \ensuremath{\mathcal{T}^{#1}_{#2}}
}
\DeclareDocumentCommand{\condTyp}{ O{\empty} O{\empty} m m m}{
  \ensuremath{\mathcal{T}^{#1}_{[#3|#4]#2}(#5)}
}
\DeclareDocumentCommand{\condType}{ O{\empty} m m m}{
  \ensuremath{\mathcal{T}^{#1}_{#2|#3}(#4)}
}

\DeclareDocumentCommand{\ent}{ O{\empty} m}{
  \ensuremath{\mathrm{H}
    \ifthenelse{\equal{#1}{big}}{\big(#2\big)}{
      \ifthenelse{\equal{#1}{Big}}{\Big(#2\Big)}{
        \ifthenelse{\equal{#1}{bigg}}{\bigg(#2\bigg)}{
          \ifthenelse{\equal{#1}{Bigg}}{\Bigg(#2\Bigg)}{
            \ifthenelse{\equal{#1}{normal}}{(#2)}{\left(#2\right)}
          }
        }
      }
    }
  }
}

\DeclareDocumentCommand{\entPhi}{ O{\empty} m}{
  \ensuremath{\mathrm{H}_{\phi}{}%
    \ifthenelse{\equal{#1}{big}}{\big(#2\big)}{
      \ifthenelse{\equal{#1}{Big}}{\Big(#2\Big)}{
        \ifthenelse{\equal{#1}{bigg}}{\bigg(#2\bigg)}{
          \ifthenelse{\equal{#1}{Bigg}}{\Bigg(#2\Bigg)}{
            \ifthenelse{\equal{#1}{normal}}{(#2)}{\left(#2\right)}
          }
        }
      }
    }
  }
}

\DeclareDocumentCommand{\condEnt}{ O{\empty} m m}{
  \ensuremath{\mathrm{H}
    \ifthenelse{\equal{#1}{big}}{\big({#2\big| #3}\big)}{
      \ifthenelse{\equal{#1}{Big}}{\Big({#2\Big| #3}\Big)}{
        \ifthenelse{\equal{#1}{bigg}}{\bigg({#2\bigg| #3}\bigg)}{
          \ifthenelse{\equal{#1}{Bigg}}{\Bigg({#2\Bigg| #3}\Bigg)}{
            \ifthenelse{\equal{#1}{normal}}{({#2| #3})}{\left(#2 \middle| #3\right)}
          }
        }
      }
    }
  }
}

\DeclareDocumentCommand{\dent}{ O{\empty} m}{
  \ensuremath{\mathrm{h}
    \ifthenelse{\equal{#1}{big}}{\big(#2\big)}{
      \ifthenelse{\equal{#1}{Big}}{\Big(#2\Big)}{
        \ifthenelse{\equal{#1}{bigg}}{\bigg(#2\bigg)}{
          \ifthenelse{\equal{#1}{Bigg}}{\Bigg(#2\Bigg)}{
            \ifthenelse{\equal{#1}{normal}}{(#2)}{\left(#2\right)}
          }
        }
      }
    }
  }
}

\DeclareDocumentCommand{\condDent}{ O{\empty} m m}{
  \ensuremath{\mathrm{h}
    \ifthenelse{\equal{#1}{big}}{\big({#2\big| #3}\big)}{
      \ifthenelse{\equal{#1}{Big}}{\Big({#2\Big| #3}\Big)}{
        \ifthenelse{\equal{#1}{bigg}}{\bigg({#2\bigg| #3}\bigg)}{
          \ifthenelse{\equal{#1}{Bigg}}{\Bigg({#2\Bigg| #3}\Bigg)}{
            \ifthenelse{\equal{#1}{normal}}{({#2| #3})}{\left(#2 \middle| #3\right)}
          }
        }
      }
    }
  }
}

\newcommand{\binEntOp}{\mathrm{H}}

\DeclareDocumentCommand{\binEnt}{ O{\empty} m}{
  \ensuremath{\binEntOp{
      \ifthenelse{\equal{#2}{\empty}}{}{
        \ifthenelse{\equal{#1}{big}}{\big({#2}\big)}{
          \ifthenelse{\equal{#1}{Big}}{\Big({#2}\Big)}{
            \ifthenelse{\equal{#1}{bigg}}{\bigg({#2}\bigg)}{
              \ifthenelse{\equal{#1}{Bigg}}{\Bigg({#2}\Bigg)}{
                \ifthenelse{\equal{#1}{normal}}{({#2})}{\left(#2\right)}
              }
            }
          }
        }
      }
    }
  }
}

\DeclareDocumentCommand{\binEntInv}{ O{\empty} m}{
  \ensuremath{\binEntOp^{-1}{
      \ifthenelse{\equal{#2}{\empty}}{}{
        \ifthenelse{\equal{#1}{big}}{\big({#2}\big)}{
          \ifthenelse{\equal{#1}{Big}}{\Big({#2}\Big)}{
            \ifthenelse{\equal{#1}{bigg}}{\bigg({#2}\bigg)}{
              \ifthenelse{\equal{#1}{Bigg}}{\Bigg({#2}\Bigg)}{
                \ifthenelse{\equal{#1}{normal}}{({#2})}{\left(#2\right)}
              }
            }
          }
        }
      }
    }
  }
}

\def\barcirc{\mathrel{\barcirci}}
\def\barcirci{{%
    \setbox0\hbox{\ensuremath{\relbar\!\!\relbar}}%
    \rlap{\hbox to \wd0{\hss\ensuremath{\circ}\hss}}\box0
}}

\newcommand{\mkv}{\ensuremath{\barcirc}}

\DeclareDocumentCommand{\uniform}{ O{\empty} m}{
  \ensuremath{\mathcal{U}
    \ifthenelse{\equal{#1}{big}}{\big({#2}\big)}{
      \ifthenelse{\equal{#1}{Big}}{\Big({#2}\Big)}{
        \ifthenelse{\equal{#1}{bigg}}{\bigg({#2}\bigg)}{
          \ifthenelse{\equal{#1}{Bigg}}{\Bigg({#2}\Bigg)}{
            \ifthenelse{\equal{#1}{normal}}{({#2})}{\left(#2\right)}
          }
        }
      }
    }
  }
}

%% file: custom_variables.tex
\newcommand{\set}[1]{\ensuremath{\mathcal{#1}}}

\newcommand{\AAA}{\ensuremath{\set{A}}}
\newcommand{\BBB}{\ensuremath{\set{B}}}

\newcommand{\III}{\ensuremath{\set{I}}}

\newcommand{\MMM}{\ensuremath{\set{M}}}

\newcommand{\QQQ}{\ensuremath{\set{Q}}}
\newcommand{\RRR}{\ensuremath{\set{R}}}
\newcommand{\SSS}{\ensuremath{\set{S}}}

\newcommand{\eps}{\ensuremath{\varepsilon}}

\newcommand{\dist}{\ensuremath{\mathrm{d}}}

\newcommand{\RRRib}{\ensuremath{\RRR_{\mathrm{IB}}}}

\newcommand{\thetafkt}[1][\empty]{\ensuremath{\hat{\theta}(\ifthenelse{\equal{#1}{\empty}}{\rho,\aii,\bii}{#1})}}

\newcommand{\aii}{\ensuremath{\alpha}}

\newcommand{\bii}{\ensuremath{\beta}}

\newcommand{\CST}[1][\empty]{\ensuremath{C_{\ifthenelse{\equal{#1}{\empty}}{\aii,\bii}{#1}}}}

%% file: tikzsettings.tex
\ifdefined\tikzextpath
\usetikzlibrary{external}
\fi

\usetikzlibrary{arrows}
\usetikzlibrary{arrows.meta}
\usetikzlibrary{decorations}
\usetikzlibrary{backgrounds}
\usetikzlibrary{fit}
\usetikzlibrary{calc}
\usetikzlibrary{intersections}
\usetikzlibrary{decorations}
\usetikzlibrary{spy}
\usetikzlibrary{positioning}
\usetikzlibrary{shapes}
\usetikzlibrary{graphs}
\usetikzlibrary{shapes.multipart}
\usetikzlibrary{automata}
\usetikzlibrary{fixedpointarithmetic}
\usetikzlibrary{decorations.markings}

\definecolor{myred}{RGB}{128, 0, 0}
\definecolor{myblue}{RGB}{0, 0, 128}

\tikzstyle{block} = [draw,fill=RoyalBlue!30,minimum size=2em,rounded corners=2mm]
\tikzstyle{symb}=[]
\def\radius{.8mm} 
\tikzstyle{branch}=[fill,shape=circle,minimum size=3pt,inner sep=0pt]
\tikzstyle{s}=[shift={(0mm,\radius)}]

\makeatletter
\newcommand{\includetikz}[1]{%
  \ifdefined\tikzexternalize%
  \filename@parse{#1}%
  \tikzsetnextfilename{\filename@base}%
  \fi%
  \input{#1.tikz}%
}
\makeatother

%% file: inner-bound.tikz
\begin{tikzpicture}
  [rvsketch
  ]  \matrix{
        \node [symb] (x) {$\mathllap{\rv y \mkv {}}\rv x$}; \&
    \node [symb] (u) {$\rv u$}; \\
        \node [symb] (hx) {$\hrv x$}; \&
    \node [symb] (wu) {$\wrv u$}; \\
  };

    \draw (x) -- node [draw,circle] {} (u);

  \draw [->] (x) -- node [anchor=east] {$g(\wc)$} (hx);
  \draw (hx) -- node [draw,circle] {} (wu);

  \node [rotate=90] (approx) at ($(u)!0.5!(wu)$) {$\approx$};
  \end{tikzpicture}%

%% file: outer-bound.tikz
\begin{tikzpicture}
  [rvsketch
  ]  \matrix{
        \node [symb] (x) {$\mathllap{\rvt z \rvt y \mkv {}}\rvt x$}; \&
    \node [symb] (f) {$f(\rvt x)$}; \\
        \node [symb] (hx) {$\hrvt x$}; \&
    \node [symb] (g) {$g(\hrvt x)$}; \\
  };

    \draw [->] (x) -- (f);

  \draw [->] (x) -- node [anchor=east] {$\fkt[x]^n(\wc)$} (hx);
  \draw [->] (hx) -- node [anchor=south] {$g(\wc)$} (g);

  \node [rotate=90] (approx) at ($(f)!0.5!(g)$) {$\approx$};
  \end{tikzpicture}%

%% file: generalIB.bbl
\begin{thebibliography}{10}
\providecommand{\url}[1]{#1}
\csname url@samestyle\endcsname
\providecommand{\newblock}{\relax}
\providecommand{\bibinfo}[2]{#2}
\providecommand{\BIBentrySTDinterwordspacing}{\spaceskip=0pt\relax}
\providecommand{\BIBentryALTinterwordstretchfactor}{4}
\providecommand{\BIBentryALTinterwordspacing}{\spaceskip=\fontdimen2\font plus
\BIBentryALTinterwordstretchfactor\fontdimen3\font minus
  \fontdimen4\font\relax}
\providecommand{\BIBforeignlanguage}[2]{{%
\expandafter\ifx\csname l@#1\endcsname\relax
\typeout{** WARNING: IEEEtran.bst: No hyphenation pattern has been}%
\typeout{** loaded for the language `#1'. Using the pattern for}%
\typeout{** the default language instead.}%
\else
\language=\csname l@#1\endcsname
\fi
#2}}
\providecommand{\BIBdecl}{\relax}
\BIBdecl

\bibitem{Tishby2000Information}
\BIBentryALTinterwordspacing
N.~Tishby, F.~C. Pereira, and W.~Bialek, ``The information bottleneck method,''
  in \emph{Proc. 37\textsuperscript{th} Annual Allerton Conference on
  Communication, Control, and Computing}, Monticello, IL, Sep. 1999, pp.
  368--377.
\BIBentrySTDinterwordspacing

\bibitem{Slonim2000Document}
\BIBentryALTinterwordspacing
N.~Slonim and N.~Tishby, ``Document clustering using word clusters via the
  information bottleneck method,'' in \emph{Proceedings of the 23rd Annual
  International ACM SIGIR Conference on Research and Development in Information
  Retrieval}, ser. SIGIR '00, ACM.\hskip 1em plus 0.5em minus 0.4em\relax New
  York, NY, USA: ACM, 2000, pp. 208--215.
\BIBentrySTDinterwordspacing

\bibitem{Gordon2003Applying}
S.~Gordon, H.~Greenspan, and J.~Goldberger, ``Applying the information
  bottleneck principle to unsupervised clustering of discrete and continuous
  image representations,'' in \emph{Computer Vision, 2003. Proceedings. Ninth
  IEEE International Conference on}, Oct 2003, pp. 370--377 vol.1.

\bibitem{Schneidman2002Analyzing}
E.~Schneidman, N.~Slonim, N.~Tishby, R.~deRuyter~van Steveninck, and W.~Bialek,
  ``Analyzing neural codes using the information bottleneck method,''
  \emph{Proc. of Advances in Neural Information Processing System (NIPS-13)},
  2002.

\bibitem{Zeitler2009quantizer}
G.~Zeitler, R.~K\"otter, G.~Bauch, and J.~Widmer, ``{On quantizer design for
  soft values in the multiple-access relay channel},'' in \emph{Proc. IEEE ICC
  2009}, Jun. 2009.

\bibitem{Zeitler2012Low}
G.~Zeitler, A.~Singer, and G.~Kramer, ``Low-precision {A/D} conversion for
  maximum information rate in channels with memory,'' \emph{IEEE Trans.
  Communications}, vol.~60, no.~9, pp. 2511--2521, 9 2012.

\bibitem{Winkelbauer2012Joint}
A.~Winkelbauer and G.~Matz, ``Joint network-channel coding for the asymmetric
  multiple-access relay channel,'' in \emph{Proc. IEEE ICC 2012}, Jun. 2012,
  pp. 2485--2489.

\bibitem{Winkelbauer2014rate}
A.~Winkelbauer, S.~Farthofer, and G.~Matz, ``The rate-information trade-off for
  {G}aussian vector channels,'' in \emph{IEEE Int. Symp. Information Theory},
  Honolulu, HI, USA, 6 2014.

\bibitem{Winkelbauer2015quantization}
A.~Winkelbauer and G.~Matz, ``On quantization of log-likelihood ratios for
  maximum mutual information,'' in \emph{Proc. IEEE SPAWC}, 6 2015, pp.
  316--320.

\bibitem{Pichler2015Distributed}
G.~Pichler, P.~Piantanida, and G.~Matz, ``Distributed information-theoretic
  biclustering of two memoryless sources,'' in \emph{Proc.
  53\textsuperscript{rd} Annual Allerton Conference on Communication, Control,
  and Computing}, Monticello, IL, Sep. 2015, pp. 426--433.

\bibitem{Courtade2014Multiterminal}
T.~A. Courtade and T.~Weissman, ``Multiterminal source coding under logarithmic
  loss,'' \emph{{IEEE} Trans. Inf. Theory}, vol.~60, no.~1, pp. 740--761, Jan.
  2014.

\bibitem{Chechik2005Information}
G.~Chechik, A.~Globerson, N.~Tishby, and Y.~Weiss, ``Information bottleneck for
  gaussian variables,'' \emph{Journal of Machine Learning Research}, vol.~6,
  pp. 165--188, Jan. 2005.

\bibitem{Slonim2000Agglomerative}
\BIBentryALTinterwordspacing
N.~Slonim and N.~Tishby, ``Agglomerative information bottleneck,'' in
  \emph{Advances in Neural Information Processing Systems 12}, S.~Solla,
  T.~Leen, and K.~M\"{u}ller, Eds.\hskip 1em plus 0.5em minus 0.4em\relax MIT
  Press, 2000, pp. 617--623.
\BIBentrySTDinterwordspacing

\bibitem{Kurkoski2017Relationship}
B.~M. Kurkoski, ``On the relationship between the kl means algorithm and the
  information bottleneck method,'' in \emph{Proc. 11th International ITG
  Conference on Systems, Communications and Coding (SCC)}, Hamburg, Germany, 2
  2017, pp. 1--6.

\bibitem{Westover2008Achievable}
M.~B. Westover and J.~A. O'Sullivan, ``Achievable rates for pattern
  recognition,'' \emph{{IEEE} Trans. Inf. Theory}, vol.~54, no.~1, pp.
  299--320, Jan. 2008.

\bibitem{Ahlswede1986Hypothesis}
R.~Ahlswede and I.~Csiszár, ``Hypothesis testing with communication
  constraints,'' \emph{{IEEE} Trans. Inf. Theory}, vol.~32, no.~4, pp.
  533--542, Jul. 1986.

\bibitem{Han1987Hypothesis}
T.~S. Han, ``Hypothesis testing with multiterminal data compression,''
  \emph{{IEEE} Trans. Inf. Theory}, vol.~33, no.~6, pp. 759--772, Nov. 1987.

\bibitem{Gray1990Entropy}
\BIBentryALTinterwordspacing
R.~M. Gray, \emph{Entropy and Information Theory}, 1st~ed.\hskip 1em plus 0.5em
  minus 0.4em\relax Springer-Verlag, 2013.
\BIBentrySTDinterwordspacing

\bibitem{Klenke2013Probability}
A.~Klenke, \emph{Probability Theory}, ser. Universitext.\hskip 1em plus 0.5em
  minus 0.4em\relax Springer-Verlag GmbH, 2013.

\bibitem{Csiszar2011Information}
I.~Csiszár and J.~Körner, \emph{Information Theory: Coding Theorems for
  Discrete Memoryless Systems}.\hskip 1em plus 0.5em minus 0.4em\relax
  Cambridge University Press, Aug. 2011.

\end{thebibliography}
